\documentclass[11pt]{article}

\usepackage[margin=1in]{geometry}
\usepackage{palatino}
\usepackage{mathpazo}
\usepackage{graphicx}

\usepackage[utf8]{inputenc}
\usepackage[bookmarks]{hyperref}
\hypersetup{colorlinks=true,citecolor=blue,linkcolor=blue,filecolor=blue,urlcolor=blue}
\usepackage{amsmath,amssymb,amsthm}
\allowdisplaybreaks[4]
\usepackage{dsfont}

\usepackage[affil-it]{authblk}

\usepackage{cleveref}

\usepackage{enumerate}

\usepackage{mathtools}

\def\>{\rangle}
\def\<{\langle}
\def\be{\begin{equation}}
\def\ee{\end{equation}}

\newtheorem{theorem}{Theorem}
\newtheorem{lemma}[theorem]{Lemma}

\newtheorem{corollary}[theorem]{Corollary}
\newtheorem{definition}[theorem]{Definition}

\usepackage{tikz}
\usepackage{pgfplots}
\usepackage[font={small},margin=1in]{caption}
\usepackage{subcaption}

\newcommand{\cB}{\mathcal{B}}
\newcommand{\cC}{\mathcal{C}}
\newcommand{\cD}{\mathcal{D}}

\newcommand{\cH}{\mathcal{H}}

\newcommand{\cJ}{\mathcal{J}}
\newcommand{\cK}{\mathcal{K}}

\newcommand{\cM}{\mathcal{M}}
\newcommand{\cN}{\mathcal{N}}

\newcommand{\one}{\mathds{1}}
\newcommand{\eps}{\varepsilon}

\DeclareMathOperator{\tr}{tr}

\DeclareMathOperator{\id}{id}

\newcommand{\sumi}{\sum\nolimits}

\newcommand{\ox}{\otimes}
\newcommand{\dn}{\diamond}
\newcommand{\bp}{\mathbf{p}}
\newcommand{\bs}{\mathbf{s}}
\newcommand{\bq}{\mathbf{q}}
\newcommand{\ip}[2]{\langle #1 , #2\rangle}
\DeclareMathOperator{\dg}{dg}

\usepackage[backend=bibtex8,firstinits=true,doi=false,isbn=false,url=false,maxbibnames=5,noerroretextools]{biblatex}
\renewbibmacro{in:}{}
\usepackage{autonum}

\title{\bf Quantum and private capacities of low-noise channels \\[0ex]}
\author[a,b]{Felix Leditzky\thanks{Email: \texttt{felix.leditzky@jila.colorado.edu}}}
\author[c]{Debbie Leung\thanks{Email: \texttt{wcleung@uwaterloo.ca}}}
\author[a,b,d]{Graeme Smith\thanks{Email: \texttt{gsbsmith@gmail.com}}}
\affil[a]{\small JILA, University of Colorado/NIST, 440 UCB, Boulder, CO 80309, USA}
\affil[b]{\small Center for Theory of Quantum Matter, University of Colorado, Boulder, Colorado 80309, USA}
\affil[c]{\small Institute for Quantum Computing, University of Waterloo, Waterloo, Ontario, Canada, N2L 3G1}
\affil[d]{\small Department of Physics, University of Colorado, 390 UCB, Boulder, CO 80309, USA}

\date{~\vspace*{-3ex} \\ \today} 

\begin{document}
\maketitle

\begin{abstract}
We determine both the quantum and the private capacities of low-noise
quantum channels to leading orders in the channel's distance to the
perfect channel.  It has been an open problem for more than 20 years 
to determine the
capacities of some of these low-noise channels such as the
depolarizing channel.

We also show that both capacities are equal to the single-letter
coherent information of the channel, again to leading orders.  We thus find
that, in the low noise regime, super-additivity and degenerate codes
have negligible benefit for the quantum capacity, and shielding does
not improve the private capacity beyond the quantum capacity, in
stark contrast to the situation when noisier channels are considered.  
\end{abstract}

\setlength{\parskip}{1ex} \setlength{\parindent}{0em}

\section{Introduction}

Any point-to-point communication link can be modeled as a quantum
channel $\cN$ from a sender to a receiver.  Of fundamental interest
are the {\em capacities} of $\cN$ to transmit data of various types such as quantum, private, or classical data.
Informally, the
capacity of $\cN$ to transmit a certain type of data is the optimal
rate at which that data can be transmitted with high fidelity given an
asymptotically large number of uses of $\cN$.  Capacities of a channel
quantify its value as a communication resource.

In the classical setting, the capacity of a classical channel $\cN$ to
transmit classical data is given by Shannon's noisy coding theorem
\cite{Shannon48}.  While operationally, the capacity-achieving error
correcting codes may have increasingly large block lengths, the
capacity can be expressed as a {\em single letter formula}: it is the
maximum correlation between input and output that can be generated
with a {\em single} channel use, where correlation is measured by the mutual
information.

In the quantum setting, the capacity of a quantum channel $\cN$ to
transmit quantum data, denoted $Q(\cN)$, is given by the LSD theorem
\cite{Lloyd97,Shor02,Devetak05}.  A capacity expression is found, but
it involves a quantity optimized over an {\em unbounded} number of
uses of the channel.  This quantity, when optimized over $n$ channel
uses, is called the $n$-shot coherent information.  Dividing the
$n$-shot coherent information by $n$ and taking the limit $n \to
\infty$ gives the capacity.
For special channels called {\em degradable} channels, the coherent
information is {\em weakly additive}, meaning that the $n$-shot
coherent information {\em is} $n$ times the $1$-shot coherent
information \cite{Shor-Devetak-05}, hence the capacity is the $1$-shot
coherent information and can be evaluated in principle.
In general, the coherent information can be {\em superadditive},
meaning that the $n$-shot coherent information can be more than $n$
times the $1$-shot coherent information, thus the optimization over $n$
is necessary \cite{DiVincenzoSS98}.  Consequently, there is no general
algorithm to compute the capacity of a given channel.
Furthermore, the $n$-shot coherent information can be positive for
some small $n$ while the $1$-shot coherent information is zero
\cite{DiVincenzoSS98}.  
Moreover, given any $n$, there is a channel whose $n$-shot coherent
information is zero but whose quantum capacity is positive
\cite{Cubitt2015unbounded}.
Thus we do not have a general method to determine if a given
channel has positive quantum capacity.

Even for the qubit depolarizing channel, which acts as 
$\cD_p(\rho) = (1-\frac{4p}{3})\,\rho + \frac{4p}{3} \frac{I}{2}$, 
our understanding of the quantum capacity is limited.  For $p=0$ the channel is perfect, so we have $Q(\cD_0) = 1$, while for $p \geq 1/4$, we know that $Q(\cD_{p}) = 0$ \cite{BDEFMS98}. However, for $0<p<1/4$ the quantum capacity of $\cD_p$ is unknown despite substantial effort (see e.g.~\cite{SS07,Fern08,LW15}).
For $p\approx 0.2$, communication rates higher than the
$1$-shot coherent information are achievable \cite{DiVincenzoSS98,SS07,Fern08}, but even the threshold value of $p$ where the capacity goes to zero is unknown.
For $p$ close to zero, the best lower bound for $Q(\cD_p)$ is the one-shot coherent information.
In this regime, the continuity bound developed in \cite{Leung-Smith-08} is
insufficient to constrain the quantum capacity of $\cD_p$ to 
the leading order in $p$, and 
while various other upper bounds exist, they all differ from the one-shot coherent information by $O(p)$.
Recently, a numerical upper bound on the capacity of the
low-noise depolarizing channel \cite{SSWR15} was found to be very close
to the $1$-shot coherent information.  Meanwhile, the complementary
channel for the depolarizing channel for any $p>0$ is found to always
have positive capacity \cite{LW15}, which renders several techniques
inapplicable, including those described in \cite{Watanabe2012} or a
generalization of degradability to ``information degradability'' \cite{CLS17}.

In this paper, we consider the quantum capacity of ``low-noise quantum
channels'' that are close to the identity channel, and investigate how
close the capacity is to the $1$-shot coherent information.  It has been unclear whether we should 
expect substantial nonadditivity of coherent information for such channels.
On the one hand, all known degenerate codes that provide a boost to quantum capacity first encode one logical qubit  into a small number of physical qubits, which
incurs a significant penalty in rate.  This would seem to preclude any benefit in the regime where
the $1$-shot coherent information is already quite high.  On the other hand, we have no effective methods for
evaluating the $n$-letter coherent information for large $n$, and there may well exist new types of coding strategies that incur no such penalty in the large $n$ regime.

We prove in this paper that to linear order in the noise parameter,
the quantum capacity of any low-noise channel {\em is} its $1$-shot
coherent information (see Theorem \ref{lem:qbound2}).  Consequently, degenerate
codes cannot improve the rates of these channels up to the same order.
For the special cases of the qubit depolarizing channel, the mixed
Pauli channel and their qudit generalizations, we show that the
quantum capacity and the $1$-shot coherent information agree to even
higher order (see Theorem \ref{cor:qpcapdepol}).

Our findings extend to the private capacity $P(\cN)$ of a quantum channel $\cN$,
which displays similar complexities to the quantum capacity.
The private capacity is equal to the regularized private information
\cite{Devetak05}, but the private information is not additive
(\cite{SRS08,Li09,SS09}).
In \cite{HHHO05}, the private capacity, which is never smaller than
the quantum capacity, is found to be positive for some channels that
have no quantum capacity.  The authors also characterize the type of
noise that hurts quantum transmission and that can be ``shielded''
from corrupting private data.
In \cite{LLSS14}, channels are found with almost no quantum capacity 
but maximum private capacity.  
Meanwhile, for degradable channels, the private capacity is
again equal to the $1$-shot coherent information \cite{Smith08}.  
This coincidence of $P$ and $Q$ for degradable channels means that our
findings for the quantum capacity can be carried over to private
capacity fairly easily.
In the low-noise regime the private capacity is also equal to the
$1$-shot coherent information to linear order in the noise parameter,
and is equal to the quantum capacity in the same order (see Theorems \ref{lem:qbound2} and \ref{cor:qpcapdepol}).  
Consequently, shielding provides little benefit.

Our results follow closely the approach in \cite{SSWR15}.  Consider a
channel $\cN$ and its complementary channel $\cN^c$. The channel $\cN$ is
degradable if there is another channel $\cM$ (called a degrading map)
such that $\cM \circ \cN = \cN^c$.  Instead of measuring how close
$\cN$ is to some degradable channel, \cite{SSWR15} considers how close
$\cM \circ \cN$ can be to $\cN^c$ when optimizing over all channels $\cM$, a measure we call the degradability
parameter of $\cN$. Furthermore, this distance between $\cN^c$ and 
$\cM\circ\cN$ as well as the best approximate degrading map $\cM$ can be
obtained via semidefinite programming.  Continuity results, relative
to the case as if $\cN$ is degradable, can then be obtained similarly
to \cite{Leung-Smith-08}.
This new bound in \cite{SSWR15} limits the difference between the
$1$-shot coherent information and the quantum capacity to $O(\eta \log
\eta)$ where $\eta$ is the degradability parameter.  Note that
$\eta \log \eta$ does not have a finite slope at $\eta = 0$ but it 
goes to zero faster than $\eta^b$ for any $b<1$.  While this method
does not yield explicit upper bounds, once a channel of interest is
fixed, it is fairly easy to evaluate the degradability parameter (via
semidefinite programming) and the resulting capacity bounds
numerically.

The primary contribution in this paper is an analytic proof of a
surprising fact that, for low-noise channels whose diamond-norm
distance to being noiseless is $\eps$, the degradability parameter
$\eta$ grows at most as fast as $O(\eps^{1.5})$, rendering the
gap $O(\eta \log \eta)$ between the $1$-shot coherent
information and the quantum or private capacity only sublinear in
$\eps$ (see Theorem \ref{thm:complementary-degrading}). 
For the qubit depolarizing channel and its various generalizations, 
we improve the analytic bound of $\eta$ to $O(\eps^{2})$ (see Theorem \ref{thm:depol-degradability-parameter}).  
Furthermore, for both results, we provide constructive approximate degrading maps and 
explain why they work well.  

The rest of the paper is structured as follows. 
In Section \ref{sec:background}, we explain both our notations and prior
results relevant to our discussion.  We present our results for a
general low-noise channel in Section \ref{sec:general} and for the
depolarizing channel in Section \ref{sec:depol-channel} (with the various
generalizations in \Cref{sec:gen-pauli-channels}).  We
conclude with some extensions and implications of our results in Section \ref{sec:conclude}.

\section{Background} 
\label{sec:background}

In this paper we only consider finite-dimensional Hilbert spaces.
For a Hilbert space $\cH$, we denote by $\cB(\cH)$ the set of linear operators on $\cH$.
We write $M_A$ for operators defined on a Hilbert space $\cH_A$ associated with a quantum system $A$.
We denote by $I_A$ the identity operator on $\cH_A$, and by $\id_A$ the identity map on $\cB(\cH_A)$.  
We denote the dimension of $A$ by $|A| = \dim\cH_A$.  
A (quantum) state $\rho_A$ on a quantum system $A$ is a positive semidefinite operator with unit trace, that is, $\rho_A\geq 0$ and $\tr\rho_A=1$.

\subsection{Quantum channels}

Any point-to-point communication link that transmits an input to
an output state can be modeled as a \emph{quantum channel}.
Mathematically, a quantum channel is a linear, completely positive,
and trace-preserving map $\cN$ from $\cB(\cH_A)$ to $\cB(\cH_B)$.  
We often use the shorthand $\cN\colon A\to B$.  
If a channel $\cN' \colon A \to B'$ acts as $\cN$ followed by an isometry from $B$ to $B'$, then we call $\cN'$ \emph{equivalent} to $\cN$. 
This is an equivalence relation on the set of quantum channels, and for a given channel $\cN$ all analysis of interest in this paper applies to any channel in the equivalence class of $\cN$.
We summarize two useful representations for quantum channels in the following
\cite{Nielsen-Chuang,Wat16}.  
 
For every quantum channel $\cN\colon A\to B$ there is an \emph{isometric extension} $V\colon \cH_A\to \cH_B\ox\cH_E$ for some auxiliary Hilbert space $\cH_E$ associated with some \emph{environment} system $E$, such that $\cN(\rho_A) = \tr_E(V\rho_A V^\dagger)$ \cite{Stinespring}.  
The isometric extension is also called the Stinespring dilation, and this representation is sometimes called the unitary representation.  
Note that the isometric extension is unique up to {\em left} multiplication by some unitary on $E$, and this degree of freedom does not affect most analysis of interest.  
Physically, the isometric extension distributes the input to $B$ and $E$ jointly, and the channel $\cN$ is noisy because the information in $E$ is no longer accessible.
The \emph{complementary channel} $\cN^c\colon A\to E$ of $\cN$ with respect to a specific isometric extension $V$ is defined by $\cN^c(\rho_A) \coloneqq \tr_B(V\rho_A V^\dagger)$.  
Note that all channels complementary to $\cN$ are equivalent.  

The second representation we use is the \emph{Choi-Jamio\l kowski isomorphism}.  It is a bijection $\cJ$ from the set of quantum channels $\cN \colon A \to B$ to the set of positive operators $\tau_{A'B} \in \cB(\cH_{A'}\ox\cH_B)$ satisfying $\tr_B \tau_{A'B} = I_{A'}$, given by:
\begin{align}
\cJ(\cN) \coloneqq (\id_{A'}\ox \cN)(\gamma_{A'A}),
\end{align}
where $|\gamma\rangle_{A'A}\coloneqq \sum_i |i\rangle_{A'}\ox|i\rangle_{A}$ is proportional to the maximally entangled state on $A'A$, the system $A'$ is isomorphic to $A$, and $\gamma_{A'A} = 
|\gamma\>\<\gamma|_{A'A}$.
The inverse of $\cJ$ applied to $\tau_{A'B}$ yields a channel whose action on an operator $\rho_A$ defined on a system $A$ is given by
\begin{align}
\left( \cJ^{-1}(\tau_{A'B}) \right) (\rho_A) = \tr_{A'} \left(\tau_{A'B}\left( (\rho^T)_{A'} \ox\one_B\right)\right),
\end{align}
where $T$ denotes transposition with respect to the basis $\lbrace |i\rangle\rbrace$ that defines $|\gamma\rangle$.  The operator $\cJ(\cN)$ is called the Choi matrix of $\cN$. It is uniquely determined by $\cN$ and it is basis-dependent. 
The rank of $\cJ(\cN)$ is called the Choi rank of $\cN$.  It is the
minimum dimension of the environment $E$ for an isometric extension of
$\cN$.  It is basis-independent, and independent of the choice of the 
isometric extension.  

\subsection{The diamond norm and the continuity of the isometric extension} 
\label{sec:diamond}

For an operator $M\in\cB(\cH)$ we define the \emph{trace norm} $\|M\|_1$, the \emph{operator norm} $\|M\|_\infty$, and the max norm $\|M\|_{\rm \max}$ as follows:
\begin{align}
\|M\|_1 &\coloneqq \tr\sqrt{M^\dagger M}\, , \\
\|M\|_\infty &\coloneqq \max \left\lbrace \sqrt{\langle\psi|M^\dagger M|\psi\rangle} \colon |\psi\rangle\in\cH, \langle\psi|\psi\rangle = 1 \right\rbrace \, , \\
\|M\|_{\rm \max} &\coloneqq \max_{i,j} |M_{i,j}|\, .
\end{align} 
Note that the max norm is basis-dependent, unlike the trace and the operator norms. 

We now discuss the distance measure we use for channels.  
For a linear map $\Phi\colon \cB(\cH)\to\cB(\cK)$ between Hilbert spaces $\cH$ and $\cK$, we define its \emph{diamond norm} by
\begin{align}
\|\Phi\|_\dn \coloneqq \max \lbrace \|(\id_{\mathbb{C}^n} \ox \Phi)(M) \|_1 \colon M\in\cB(\mathbb{C}^n\ox\cH), \|M\|_1 = 1, n\in\mathbb{N} \rbrace,\label{eq:diamond-norm}
\end{align}
where $\id_{\mathbb{C}^n}$ denotes the identity map on $\mathbb{C}^n$.
It suffices to take $n$ as large as $\dim(\cH)$ in the above optimization 
so that the maximum can be attained. 
When applied to the difference of two channels $\cN_1 - \cN_2$, the diamond
norm has a simple operational meaning: it is twice the maximum of the trace
distance of the two output states $(\id_{\mathbb{C}^n} \ox \cN_i )(M)$
(for $i=1,2$) created by the two channels on the best common input $M$, in the
presence of an ancillary space $\mathbb{C}^n$.  The trace distance 
of two states in turn signifies their distinguishability \cite{Helstrom}. 
We summarize two characterizations of the diamond norm, a method to
upper bound it, and a continuity result for the isometric extension in
the rest of this subsection.

First, the diamond norm $\|\cN_1 - \cN_2\|_\dn$ of the difference of two quantum channels $\cN_1\colon A\to B$ and $\cN_2\colon A\to B$ can be computed by solving the following semidefinite 
program (SDP) \cite{Wat09}:
\begin{align}
\begin{aligned}
\text{\normalfont minimize: } & 2\mu\\
\text{\normalfont subject to: } & \tr_BZ_{AB} \leq \mu I_A\\
& Z_{AB} \geq \cJ(\cN_1) - \cJ(\cN_2)\\
& Z_{AB} \geq 0.
\end{aligned}
\label{eq:diamond-norm-sdp}
\end{align}

Second, the diamond norm of a linear map $\Theta$ (which is not necessarily trace-preserving or completely positive) can be rewritten as (\cite{Wat12}, Thm.~6):
\begin{align}
\|\Theta\|_\dn = \max_{\rho_1,\,\rho_2} \left\lbrace \left\|(\sqrt{\rho_1}\ox\one_B) \cJ(\Theta) (\sqrt{\rho_2}\ox\one_B) \right\|_1 \right\rbrace,
\label{eq:diamond-norm-rewrite}
\end{align}
where the maximum is over density matrices $\rho_1$ and $\rho_2$ on the input system $A$ of the map $\Theta$.

We prove the following technical lemma that upper bounds the diamond norm of an arbitrary linear map using the max norm of its Choi matrix:  

\begin{lemma}\label{lem:coefficients-bound}
	For a linear map $\Theta\colon A\to B$, we have $\|\Theta\|_\dn \leq |A| \, |B|^2 
	\; \|\cJ(\Theta)\|_{\max}.$
\end{lemma}

\begin{proof}
	We start with the second characterization of the diamond norm in \eqref{eq:diamond-norm-rewrite} above, and           
	let $\sigma_1$ and $\sigma_2$ be states on $A$ achieving the maximum in \eqref{eq:diamond-norm-rewrite}, such that
	\begin{align}
	\|\Theta\|_\dn &= \left\|(\sqrt{\sigma_1}\ox\one_B) \, \cJ(\Theta) \, (\sqrt{\sigma_2}\ox\one_B) \right\|_1.\label{eq:dn-trace-distance}
	\end{align}
	Recall that the Schatten norms $\|X\|_p\coloneqq \left(\tr(X^\dagger X)^{p/2}\right)^{1/p}$ satisfy the H\"{o}lder inequality $\|XY\|_1 \leq \|X\|_p\|Y\|_q$ for $p,q\in[1,\infty]$ with $1=\frac{1}{p}+\frac{1}{q}$.
	Applying this to the right-hand side of \eqref{eq:dn-trace-distance} with $p=2=q$ gives
	\begin{align}
		\|\Theta\|_\dn = \left\|(\sqrt{\sigma_1}\ox\one_B) \, \cJ(\Theta) \, (\sqrt{\sigma_2}\ox\one_B) \right\|_1 &\leq \left\|\sqrt{\sigma_1}\ox\one_B\right\|_2 \left\|\cJ(\Theta) \, (\sqrt{\sigma_2}\ox\one_B) \right\|_2\\
		&\leq \left\|\sqrt{\sigma_1}\ox\one_B\right\|_2 \left\|\cJ(\Theta)\right\|_\infty \left\|\sqrt{\sigma_2}\ox\one_B \right\|_2,
	\end{align}
	where the second inequality uses $\|XY\|_2^2 = \tr(X^\dagger X YY^\dagger) \leq \tr(X^\dagger X \|Y\|_\infty^2I) = \|Y\|_\infty^2 \|X\|_2^2$.
	We have $\|\sqrt{\sigma_i}\ox \one_B\|_2 = (\tr(\sigma_i\ox\one_B))^{1/2} = |B|^{1/2}$ for $i=1,2$ since the $\sigma_i$ are quantum states.
	Hence,
	\begin{align}
		\|\Theta\|_\dn \leq |B| \; \|\cJ(\Theta)\|_\infty.\label{eq:dn-estimate}
	\end{align}

	To bound the operator norm of $\cJ\equiv \cJ(\Theta)$ in \eqref{eq:dn-estimate}, let $n=|A||B|$ and $s_1\geq \dots \geq s_n\geq 0$ be the singular values of $J$ in descending order.
	Recall furthermore that $\|X\|_\infty^2 = \|X^\dagger X\|_\infty$.
	Then,
	\begin{align}
	\|\cJ\|_\infty^2 = \|\cJ^\dagger \cJ\|_\infty =  s_1^2 \leq \sum_{k=1}^n s_k^2 = \tr(\cJ^\dagger \cJ) = \sum_{i,\,j=1}^{n} \overline{\cJ_{ij}} \cJ_{ij} = \sum_{i,\,j=1}^{n} |\cJ_{ij}|^2 \leq n^2 \|\cJ\|^2_{\max}.
	\end{align}
	Hence, $\|\cJ(\Theta)\|_\infty \leq |A| \, |B| \; \|\cJ(\Theta)\|_{\max}$, which together with \eqref{eq:dn-estimate} proves the claim.
\end{proof}

\begin{corollary}\label{cor:choi-coefficients}
	If $\Theta$ is a linear map whose Choi matrix has coefficients $O(p^2)$ for $p$ in a neighborhood of $0$, then also $\|\Theta\|_\dn = O(p^2)$.
\end{corollary}

We also use the following continuity result for the isometric extensions for channels \cite{KSW08} : 
\begin{theorem}[\cite{KSW08}]\label{thm:stinespring-continuity}
	For quantum channels $\cN_1$ and $\cN_2$,
	\begin{align}
	\inf_{V_1,V_2} \|V_1 - V_2\|_{\infty}^2 \leq \|\cN_1 - \cN_2 \|_{\dn} \leq 2\inf_{V_1,V_2} \|V_1 - V_2\|_\infty,
	\end{align}
	where the infimum is over isometric extensions $V_i$ of $\cN_i$ for $i=1,2$, respectively.
\end{theorem}

A simple consequence of \Cref{thm:stinespring-continuity} is that for two quantum channels that are close in the diamond norm, their complementary channels can be made similarly close to one another:
\begin{corollary}\label{cor:complementary-channel-norm}
	Let $\cN_1, \cN_2$ be quantum channels with $\|\cN_1 - \cN_2\|_\dn \leq \eps$ for some $\eps\in[0,2]$.
	Then for any complementary channel $\cN_1^c$ of $\cN_1$, there exists a complementary channel $\cN_2^c$ of $\cN_2$ such that
	\begin{align}
	\|\cN_1^c - \cN_2^c\|_\dn \leq 2\sqrt{\eps}.
	\end{align}
\end{corollary} 

\subsection{Quantum and private capacities and approximate degradability}
\label{sec:appdeg}

The \emph{coherent information} of a state $\rho_A$ through a channel
$\cN \colon A \to B$ is defined as
\begin{equation}
I_c(\rho_A;\cN)
\coloneqq S(\cN(\rho_A)) - S(\cN^c(\rho_A)) \,,
\end{equation}
where $S(\sigma) \coloneqq {-}\tr\sigma \log \sigma$ denotes the von~Neumann 
entropy of $\sigma$, and $\log$ is base $2$ throughout this paper.
Note that the coherent information is independent of the choice of the
complementary channel.
The coherent information can be interpreted as follows.  Let 
$|\psi\>_{A'A}$ be a purification of $\rho_A$ (that is, 
$\tr_{A'} |\psi\>\<\psi|_{A'A} = \rho_A$).  
Then, $I_c(\rho_A;\cN) = \frac{1}{2} \left( I(A':B) - I(A':E) \right)$ 
where $I(A':B) = S(A') + S(B) - S(A'B)$ is the quantum mutual information 
between $A'$ and $B$, and similarly for $I(A':E)$, and the mutual 
information is evaluated on the state $\id_{A'} \otimes \cN (|\psi\>\<\psi|_{A'A})$.  
The coherent information of $\cN$, also called the $1$-shot coherent information, 
is the maximum over all input states,
\begin{align}
I_c(\cN)
\coloneqq \max_{\rho_A} I_c(\rho_A;\cN) \,.
\end{align}
The $n$-shot coherent information of $\cN$ is 
defined as $I_c^{(n)}(\cN) \coloneqq I_c(\cN^{\otimes n})$, and satisfies $I_c^{(n+m)}(\cN) \geq I_c^{(n)}(\cN) + I_c^{(m)}(\cN)$ for $n,m\in\mathbb{N}$.  
The \emph{quantum capacity theorem} \cite{Lloyd97,Shor02,Devetak05}
establishes that the quantum capacity of $\cN$ is given by the following regularized formula:
\begin{align}
Q(\cN) = \lim_{n\rightarrow\infty}  \frac{1}{n} \, I_c^{(n)}(\cN) = \sup_{n\in\mathbb{N}}\frac{1}{n}\,I_c^{(n)}(\cN)\,,
\label{eq:qcap}
\end{align}
where the second equality follows from Fekete's lemma \cite{Fek23}.
In general, the regularization in \eqref{eq:qcap} is necessary, and renders the quantum capacity intractable to compute for most channels  \cite{DiVincenzoSS98,Cubitt2015unbounded}. 
However, for the class of degradable channels \cite{Shor-Devetak-05}, the formula \eqref{eq:qcap} reduces to a single-letter formula.
A channel $\cN$ with complementary channel $\cN^c$ is {\em degradable} if there is another channel $\cM$ (called a degrading
map) such that $\cM \circ \cN = \cN^c$. 
For degradable channels, 
the coherent information is {\em weakly additive}, $I_c^{(n)}(\cN) = n I_c(\cN)$ \cite{Shor-Devetak-05}.
As a result, the limit in \eqref{eq:qcap} is unnecessary, and we have
\begin{align}
Q(\cN) = I_c(\cN).
\end{align}
Moreover, for a degradable channel $\cN$ the coherent information $I_c(\rho_A; \cN)$ is concave in the input state $\rho_A$ \cite{YHD08}, and thus $I_c(\cN)$ can be efficiently computed using standard optimization techniques.

Since degradable channels have nice properties that allow us to determine their quantum capacity, we might ask if some of these properties are approximately satisfied by channels that are ``almost'' degradable.
Reference \cite{SSWR15} formalized this idea by considering how close $\cM \circ \cN$ can be made to $\cN^c$ when optimizing over the channel $\cM$. 
\begin{definition}[\cite{SSWR15}; Approximate degradability]\label{def:approximate-degradability}
	A quantum channel $\cN\colon A\to B$ with environment $E$ is called \emph{$\eta$-degradable} for an $\eta\in[0,2]$ if there exists a quantum channel $\cM\colon B\to E$ such that
	\begin{align}
	\|\cN^c - \cM\circ \cN \|_\dn \leq \eta.
	\label{eq:approx-degradability}
	\end{align}
	The \emph{degradability parameter} $\dg(\cN)$ of $\cN$ is defined as the minimal $\eta$ such that \eqref{eq:approx-degradability} holds for some quantum channel $\cM\colon B\to E$, and $\cN$ is degradable if $\dg(\cN)=0$.
\end{definition}
Note that every quantum channel is $2$-degradable, since $\|\cN_1-\cN_2\|_\dn\leq 2$ for any two quantum channels $\cN_1$ and $\cN_2$.
The SDP \eqref{eq:diamond-norm-sdp} for the diamond norm can be used to express the degradability parameter $\dg(\cN)$ of \Cref{def:approximate-degradability} as the solution of the following SDP \cite{SSWR15}:
\begin{align}
\begin{aligned}
\text{\normalfont minimize: } & 2\mu\\
\text{\normalfont subject to: } & \tr_EZ_{AE} \leq \mu I_A\\
& \tr_E Y_{BE} = \one_B\\
& Z_{AE} \geq \cJ(\cN^c) - \cJ(\cJ^{-1}(Y_{BE})\circ\cN)\\
& Z_{AE} \geq 0, Y_{BE} \geq 0.
\end{aligned}
\label{eq:approx-deg-sdp}
\end{align}
The bipartite operator $Y_{BE}$ above corresponds to the Choi matrix of the approximate degrading map $\cM\colon B\to E$.

While a degradable channel has capacity equal to the $1$-shot coherent
information, an $\eta$-degradable channel $\cN$ has capacity 
differing from the $1$-shot coherent information at most by a vanishing 
function of $\eta$:
\begin{theorem}[\cite{SSWR15}, Thm.~3.3$\,$(i); Continuity bound]\label{thm:deg-cty}
If $\cN$ is a channel with degradability parameter $\dg(\cN) = \eta$, then, 
\begin{align}
I_c(\cN) \leq Q(\cN) \leq 
I_c(\cN) + \frac{\eta}{2} \log(|E|{-}1) + \eta \log|E| 
+ h \! \left(\frac{\eta}{2}\right) 
+ \left(1+\frac{\eta}{2}\right) h\left( \frac{\eta}{2+\eta} \right),
\label{eq:cty-bdd}
\end{align}
where $h(x)\,{:}{=}-x\log x-(1{-}x)\log(1{-}x)$ is the binary 
entropy function, and $|E|$ is the Choi rank of $\cN$.
\end{theorem}

The private capacity of $\cN$, denoted by $P(\cN)$, is the capacity of
$\cN$ to transmit classical data with vanishing probability of error
such that the state in the joint environment of all channel uses has
vanishing dependence on the input.  The capacity expression is found
to be the regularized private information for $\cN$ in
\cite{Devetak05}, but the private information is not additive
(\cite{SRS08,Li09, SS09}).  For degradable channels, $P(\cN) =
I_c(\cN)$, and there is a continuity result similar to Theorem
\ref{thm:deg-cty}:
\begin{theorem}[\cite{SSWR15}, Thm.~3.3$\,$(iii) and (v) combined]\label{thm:deg-cty-p}
If $\cN$ is a channel with degradability parameter $\dg(\cN) = \eta$, then, 
\begin{align}
I_c(\cN) \leq P(\cN) \leq 
I_c(\cN) + \eta \log(|E|{-}1) + 4 \eta \log|E| 
+ 2 \, h \! \left(\frac{\eta}{2}\right) 
+ 4 \left(1+\frac{\eta}{2}\right) h\left( \frac{\eta}{2+\eta} \right),
\label{eq:cty-bdd-private}
\end{align}
where $|E|$ and $h$ are as defined in Theorem \ref{thm:deg-cty}.
\end{theorem}

\section{General low-noise channel}\label{sec:general}

Throughout this paper we focus on {\em low-noise channels}, by which we mean a channel $\cN$ that has isomorphic
input and output spaces $A$ and $B$ and approximates the noiseless (or identity)
channel in the diamond norm, 
$\| \cN - \id_A \|_\dn \leq \eps$, where $\eps>0$ is a small
positive parameter.

\subsection{Deviation of quantum capacity from 1-shot coherent information}

We start from \Cref{thm:deg-cty} in Section \ref{sec:appdeg}, which
gives a ``continuity bound'' on the difference between the quantum
capacity and the 1-shot coherent information for an arbitrary channel
$\cN$ with degradability parameter $\dg(\cN) = \eta$:
 \begin{align}
|Q(\cN) - I_c(\cN)| \leq f_1(\eta), 
\end{align}
where 
\begin{align}
f_1(\eta) \coloneqq \frac{\eta}{2} \log(|E|-1) + \eta \log|E| + h \! \left(\frac{\eta}{2}\right) + \left(1+\frac{\eta}{2}\right) h\left( \frac{\eta}{2+\eta} \right),
\label{eq:error-term}
\end{align} 
satisfying $\lim_{\eta\to 0}f_1(\eta) = 0$.

If $\cN$ is ``almost'' degradable, then $\eta \approx 0$.
To investigate the behavior of $f_1(\eta)$ in this regime, we keep the $\eta\log\eta$ terms in $f_1(\eta)$ (knowing that these will dominate for small $\eta$), and develop the rest in a Taylor series around $0$.
For example, the first binary entropy term is expanded as
\begin{align}
h \! \left(\frac{\eta}{2}\right) = \frac{1}{2}\left(1+(\ln 2)^{-1}-\log \eta\right) \eta - \frac{1}{8} (\ln 2)^{-1} \eta^2 + O(\eta^3).
\end{align}
The second binary entropy term (including the prefactor) is expanded as
\begin{align}
\left(1+\frac{\eta}{2}\right) h\left( \frac{\eta}{2+\eta} \right) = \frac{1}{2}\left(1+(\ln 2)^{-1}-\log \eta\right) \eta + \frac{1}{8} (\ln 2)^{-1} \eta^2 + O(\eta^3).
\end{align}
The quadratic contributions cancel out in an expansion of $f_1(\eta)$, and including the two linear terms in \eqref{eq:error-term} gives
\begin{align}
f_1(\eta) = -\eta\log\eta + \left( 1 + (\ln 2)^{-1} + \frac{1}{2}\log(|E|-1) + \log |E| \right)\eta + O(\eta^3).\label{eq:error-term-expanded}
\end{align}
It follows that for small $\eta$ the function $f_1(\eta)$ is dominated by $g(\eta) \,{:}{=}\, -\eta \log \eta$ 
which has infinite slope at $\eta=0$: 
\begin{align}
g'(\eta) = -\log(\eta) - \frac{1}{\ln 2}\xrightarrow{\eta \, \to \, 0} +\infty.
\label{dgdeta}
\end{align} 
Hence, without further information on $\eta$, \Cref{thm:deg-cty} does not give 
a tight approximation of the quantum capacity relative to the 1-shot coherent 
information.  

Instead, we consider a scenario in which the channel $\cN(p)$ depends
on some underlying noise parameter $p\in[0,1]$ and 
$\eta = \dg(\cN(p)) \leq c p^r$ where $r>1$ and $c$ is a constant.
(Note that $f_1(\eta)$ increases with $\eta$ for small $\eta$, 
so, it suffices to consider $\eta = c p^r$.)  
In this case, the approximation implied by \Cref{thm:deg-cty} becomes
extremely useful -- we will show that the upper bound 
$I_c(\cN(p)) + f_1(cp^r)$ on $Q(\cN(p))$ is now \emph{tangent} to the
$1$-shot coherent information $I_c(\cN(p))$.
\begin{lemma} \label{lem:fderivative}
If $r>1$ and $c$ is a constant, then $\frac{d}{dp}f_1(cp^r)\Big|_{p=0} = 0$.  
\end{lemma}
\begin{proof}
	From \eqref{dgdeta} and the chain rule, we obtain 
	\begin{align}
	\frac{d}{dp}g(cp^r) = \left(-\log(cp^r) - \frac{1}{\ln 2} \right) c rp^{r-1}
        \end{align}
	Hence, $\lim_{p\to 0} \frac{d}{dp}g(cp^r) = 0$ for any $r>1$.
        Finally, 
        \begin{align}
        \frac{d}{dp}f_1(cp^r) = \frac{d}{dp}g(cp^r) + 
        \left( 1+\left(\ln 2\right)^{-1} + \log|E| + \frac{1}{2}\log(|E|-1) \right) c rp^{r-1} + O(p^{3r{-}1}).
        \end{align}
        So, $\frac{d}{dp}f_1(cp^r)\Big|_{p=0} = 0$ as claimed.  
\end{proof}

\Cref{fig:g-derivative} and \Cref{fig:g-function} plot $\frac{d}{dp}g(cp^r)$ and $g(cp^r)$, respectively, for the values $c=1$ and $r\in\lbrace 1, 1.5, 2\rbrace$.

\begin{figure}[ht]
	\centering
	\begin{tikzpicture}
	\begin{axis}[
		scale = 1.2,
		domain=0.00001:0.1,
		axis lines = left,
		scaled ticks=false, 
		tick label style={/pgf/number format/fixed}, 
		legend cell align = left,
		grid = major,
		samples = 150,
		ymax = 6,
		xmin = 0,
		xlabel = $p$,
		every axis x label/.style={at={(current axis.right of origin)},anchor=west,right = 0.3cm}, 
		ylabel = $\frac{d}{dp}g(p^r)$,
		ylabel style={rotate=-90},
		every axis y label/.style={at={(current axis.north west)},above=0mm, left = 0.3cm},
	] 
	\addplot[smooth,mark=none, color = gray, thick] {-ln(x)/ln(2)-1.4427};
	\addplot[smooth,mark=none, color = blue, thick] {(-ln(x^1.5)/ln(2)-1.4427)*1.5*x^0.5};
	\addplot[smooth,mark=none, color = red, thick] {(-ln(x^2)/ln(2)-1.4427)*2*x};
	\legend{$r=1$,$r=1.5$,$r=2$};
	\end{axis}
	\end{tikzpicture}
	\caption{Plot of the derivative of $g(p^r)$ with respect to $p$ for $r\in \lbrace 1,1.5,2\rbrace$, where $g(\eta) =-\eta\log\eta$.}
	\label{fig:g-derivative} 
\end{figure}
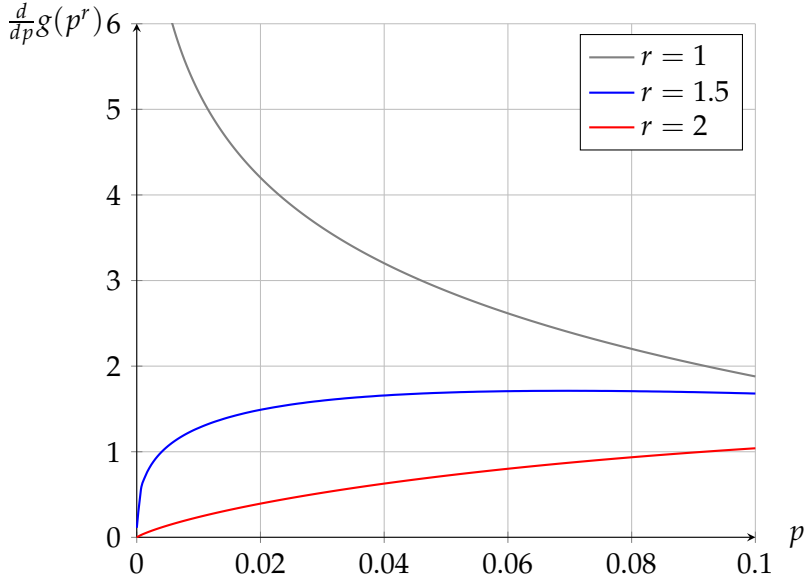

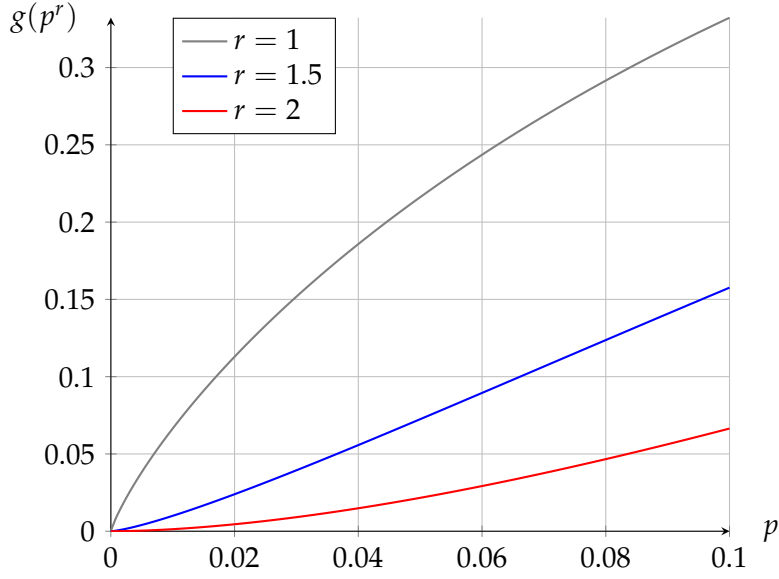
\begin{figure}[ht]
	\centering	
	\begin{tikzpicture}
	\begin{axis}[
	scale = 1.2,
	domain=0.00001:0.1,
	axis lines = left,
	scaled ticks=false, 
	tick label style={/pgf/number format/fixed}, 
	legend cell align = left,
	legend style = {at = {(0.1,1)},anchor = north west},
	grid = major,
	samples = 150,
	xlabel = $p$,
	every axis x label/.style={at={(current axis.right of origin)},anchor=west,right = 0.3cm}, 
	ylabel = $g(p^r)$,
	ylabel style={rotate=-90},
	every axis y label/.style={at={(current axis.north west)},above=0mm, left = 0.3cm},
	ytick = {0,0.05,0.1,0.15,0.2,0.25,0.3},
	]
	\addplot[smooth,mark=none, color = gray, thick] {-x*ln(x)/ln(2)};
	\addplot[smooth,mark=none, color = blue, thick] {-x^1.5*ln(x^1.5)/ln(2)};
	\addplot[smooth,mark=none, color = red, thick] {-x^2*ln(x^2)/ln(2)};
	\legend{$r=1$,$r=1.5$,$r=2$};
	\end{axis}
	\end{tikzpicture}
	\caption{Plot of the function $g(p^r)$ for $r\in \lbrace 1,1.5,2\rbrace$, where $g(\eta) =-\eta\log\eta$.}
	\label{fig:g-function}
\end{figure}

We also simplify \eqref{eq:error-term-expanded} as
\begin{lemma}
	Let $\cN = \cN(p)$ be a channel defined in terms of a noise
        parameter $p\in [0,1]$, and $\dg(\cN) \leq c p^r$ for some
        $r>1$ and some constant $c$. Then, we have
	\begin{multline}
	| Q(\cN) - I_c(\cN) | \leq 
         crp^{r{-}1} (-p\log p) \\
        + cp^r\left( -\log c + 1+\left(\ln 2\right)^{-1} 
                             + \log|E| + \frac{1}{2}\log(|E|{-}1) \right)
+ O(p^{3r}).
\end{multline} 
\label{lem:qbound1}
\end{lemma}

For the private capacity, we have the following from \Cref{thm:deg-cty-p}:
\begin{align}
|P(\cN)-I_c(\cN)| \leq f_2(\eta),
\end{align}
where
\begin{align}
f_2(\eta) & \coloneqq \eta \log(|E|{-}1) + 4 \eta \log|E| 
+ 2 \, h \! \left(\frac{\eta}{2}\right) 
+ 4 \left(1+\frac{\eta}{2}\right) h\left( \frac{\eta}{2+\eta} \right) \\
        & = -3 \eta\log\eta + \left( 3 + 3(\ln 2)^{-1} + \log(|E|-1) + 4 \log |E| \right)\eta + O(\eta^2).
\end{align}
In a similar way as above, we can derive the following:
\begin{lemma}
	Let $\cN = \cN(p)$ be a channel defined in terms of a noise
        parameter $p\in [0,1]$, and $\dg(\cN) \leq c p^r$ for some
        $r>1$ and some constant $c$. Then,
	\begin{multline}
| P(\cN) - I_c(\cN) | \leq 3cr p^{r-1}(-p\log p)\\  + c p^r (-3\log c+3+3(\ln 2)^{-1} + \log(|E|-1)+4\log|E|)  + O(p^{2r}).
\end{multline} 
\label{lem:pbound1}
\end{lemma}

Our main contribution in this paper is to note that many interesting channels $\cN$
satisfying $\|\cN-\id\|_\dn = O(p)$ have $\dg(\cN) = O(cp^r)$ for some $r>1$ and a constant $c$, 
such that \Cref{lem:qbound1} and \Cref{lem:pbound1} apply.  In the following
subsection, we prove that any channel $\cN$ which is $\eps$-close to
the identity in the diamond norm has $\dg(\cN) \leq 2 \eps^{1.5}$.
Furthermore, in \Cref{sec:paulietal} we improve the exponent to $r=2$ for the
Pauli channels.

\subsection{Degrading a low-noise channel with its complementary channel}\label{sec:low-noise-degrading}

Low-noise channels as defined at the beginning of this section are natural examples of ``almost'' degradable channels, since the identity channel is trivially degradable: its degrading map is given by its complementary channel $\id^c(\cdot) = \tr(\cdot)|0\rangle \langle 0|_E$ that outputs a fixed state to the environment.
This suggests that the same holds approximately for low-noise channels, i.e., a channel $\cN$ that is $\eps$-close to the identity should be approximately degraded by its complementary channel $\cN^c$.
Indeed, one of the main results of this paper, \Cref{thm:complementary-degrading} below, shows that a channel $\cN$ with $\|\cN-\id\|_\dn \leq \eps$ is $2\eps^{3/2}$-approximately degradable with respect to its complementary channel.
We prove later (\Cref{thm:gen-pauli-channels} in \Cref{sec:paulietal}) that the dependence on $\eps$ can be improved to quadratic order for the class of Pauli channels.

\begin{theorem}\label{thm:complementary-degrading}
	Let $\cN$ be a low-noise quantum channel, i.e., $\|\cN-\id\|_\dn \leq \eps$ for some $\eps\in[0,2]$. 
	Then $\cN$ is $2\eps^{3/2}$-approximately degradable:
	\begin{align}
	\| \cN^c - \cN^c\circ\cN \|_\dn \leq 2\eps^{3/2}.
	\end{align}
\end{theorem}

\begin{proof}
	Let $\rho$ be a quantum state achieving the maximum in the definition \eqref{eq:diamond-norm} of the diamond norm of $\cN^c - \cN^c\circ\cN$, so that   
	\begin{align}
	\| \cN^c - \cN^c\circ\cN \|_\dn &= \| \id\ox \cN^c (\rho) - \id\ox (\cN^c\circ \cN)(\rho)\|_1\\
	& = \| \id\ox \cN^c (\delta)\|_1,
	\label{eq:diamond-norm-achieved}
	\end{align}
	where $\delta\coloneqq \rho - \id\ox\cN(\rho)$.
	This is a traceless, Hermitian operator that can be expressed as
	\begin{align}
	\delta = \sumi_i \lambda_i |\psi_i\rangle\langle \psi_i|,
	\end{align}
	where $\sum_i\lambda_i = 0$ since $\delta$ is traceless, and $\sum_i |\lambda_i| \leq \eps$ since $\|\cN-\id\|_\dn \leq \eps$.
	Furthermore, we have $\tr_2\delta = \tr_2\rho - \tr_2(\id\ox\cN(\rho)) = 0$ due to $\cN$ being trace-preserving, and hence,
	\begin{align}
	\tr_2\delta = \sumi_i \lambda_i \tr_2|\psi_i\rangle\langle \psi_i| = 0.
	\label{eq:delta-partial-trace}
	\end{align}
	Using the triangle inequality for the trace distance, we bound
	\begin{align}
	\| \cN^c - \cN^c\circ\cN \|_\dn &= \| \id\ox\cN^c(\delta) \|_1\\
	&\leq \| \id\ox \cN^c(\delta) - \id \ox \id^c(\delta)\|_1 + \|\id\ox \id^c(\delta) \|_1,\label{eq:triangle-ineq}
	\end{align}
	where the complementary channel $\id^c$ of the identity map is the completely depolarizing channel
	\begin{align}
	\id^c(X) = \tr(X) |0\rangle\langle 0|,
	\end{align}
	defined in terms of some pure state $|0\rangle$ of the environment of $\cN$.
	For the first term in \eqref{eq:triangle-ineq}, we have by another application of the triangle inequality that
	\begin{align}
	\| \id\ox \cN^c(\delta) - \id \ox \id^c(\delta)\|_1 & \leq \sumi_i |\lambda_i| \| \id\ox\cN^c(|\psi_i\rangle\langle \psi_i|) - \id\ox\id^c(|\psi_i\rangle\langle \psi_i|)\|_1 \\
	&\leq 2\sqrt{\eps}\sumi_i |\lambda_i| \\ 
	&\leq 2\eps^{3/2},
	\end{align}
	where in the second inequality we used the following bound that holds for all $i$ and follows from \Cref{cor:complementary-channel-norm}:
	\begin{align}
	\| \id\ox\cN^c(|\psi_i\rangle\langle \psi_i|) - \id\ox\id^c(|\psi_i\rangle\langle \psi_i|)\|_1 \leq \| \cN^c - \id^c \|_\dn \leq 2\sqrt{\eps}.
	\end{align}
	For the second term in \eqref{eq:triangle-ineq}, we have
	\begin{align}
	\id\ox\id^c(\delta) = \sumi_i \lambda_i \tr_2(|\psi_i\rangle\langle \psi_i|) \ox |0\rangle\langle 0| = 0
	\end{align}
	by \eqref{eq:delta-partial-trace}, and hence, $\|\id\ox \id^c(\delta) \|_1 = 0$.
	This concludes the proof.	
\end{proof}

Combining \Cref{thm:complementary-degrading} with \Cref{lem:qbound1} 
and \ref{lem:pbound1}, we obtain the following main result:
\begin{theorem}\label{lem:qbound2}
	If $\|\cN-\id\|_\dn \leq \eps$ for some $\eps\in[0,2]$,
        then 
	\begin{align}
	| Q(\cN) - I_c(\cN) | &\leq 
        - 3 \eps^{1.5} \log \eps  
        + \left( \left(\ln 2\right)^{-1}
                             + \log|E| + \frac{1}{2} \log(|E|{-}1) \right) \, 2 \eps^{1.5}
        + O(\eps^{4.5}) \\
        | P(\cN) - I_c(\cN) | &\leq - 9 \eps^{1.5} \log \eps  
        + \left( 3\left(\ln 2\right)^{-1}
        + 4\log|E| + \log(|E|{-}1) \right) \, 2 \eps^{1.5}
        + O(\eps^{3}).
\end{align} 
\end{theorem}

Recall that $\eps \log \eps$ goes to zero faster than $\eps^b$ for any
$b<1$.  \Cref{lem:qbound2} thus narrows the uncertainty of both
capacities to $\eps^b$ for $b \approx 1.5$.  Furthermore, since the
channel $\cN$ can already communicate quantum data at the rate
$I_c(\cN)$ using a non-degenerate quantum error correcting code
\cite{Lloyd97,Shor02,Devetak05}, \Cref{lem:qbound2} shows that
degenerate codes only improve quantum communication rates in
$O(\eps^b)$ for $b \approx 1.5$.  Such a code also transmits
private data, and shielding cannot improve the private capacity by the
same order.

\section{Generalized Pauli channels}
\label{sec:paulietal}
In this section, we apply our results from \Cref{sec:general} to the
\emph{generalized Pauli channels} on finite dimension $d$.  This class of 
channels includes
the depolarizing channel and the $XZ$-channel acting on qubits.  The
quantum and private capacities of these channels have remained 
unknown for more than 20 years
(except for very special extreme values of the noise parameters).  

Our main result is that for generalized Pauli channels
that are $\eps$-close to the identity channel, the upper bound for the
degradability parameter in \Cref{thm:complementary-degrading} can be
improved to $O(\eps^2)$.  We show how a complementary channel with
suitably modified noise parameter can be used to achieve such  
improved approximate degrading.  

We start by introducing the \emph{generalized Pauli channels} on
finite dimension $d$.  Define the generalized Pauli $X$ and $Z$
operators on $\mathbb{C}^d$:
$$ X|i\> = |i{+}1\>, ~~Z|i\> = \omega^i |i\>\,,$$ where $\{|i\>\}$ is the
computational basis, addition is modulo $d$, and $\omega$ is a primitive
$d^{\rm th}$ root of unity.  The generalized Pauli basis is 
given by $G = \{X^k Z^l: 0 \leq k,l \leq d{-}1\}$, and the generalized 
Pauli channel has the form 
$$ \cN(\rho) = \sum_{U \in G} p_U U \rho U^\dagger \,,$$ 
where $\{p_U\}_{U \in G}$ is a probability distribution.  
The above reduces to the \emph{Pauli channel} in the special case $d=2$, 
\begin{align}
\cN_\bp(\rho) = p_0 \rho + p_1 X\rho X + p_2 Y\rho Y + p_3 Z\rho Z,
\label{eq:gen-pauli-channel}
\end{align}
where $X, Y, Z$ are the usual Pauli matrices in $\cB(\mathbb{C}^2)$, 
and $\mathbf{p}=(p_0,p_1,p_2,p_3)$ is a probability distribution.  

We first illustrate the main ideas on the simpler depolarizing channel, and then state the general result for the Pauli
channel which we prove in \Cref{sec:proof-gen-pauli-channels}.  Generalization to
higher dimensions can be done by expressing the channel input in the
basis $G$, and noting that the generalized Pauli channel acts diagonally
in this basis.  The derivation is straightforward and left as an
exercise for the interested readers.

\subsection{Depolarizing channel}\label{sec:depol-channel}
We first consider the \emph{qubit depolarizing channel} with error $p$: 
\begin{align}
\cD_p (\rho) \coloneqq (1-p)\rho + \frac{p}{3}(X\rho X + Y\rho Y + Z\rho Z)\quad \text{for }p\in[0,1],\label{eq:depolarizing-channel}
\end{align}
which corresponds to setting $p_1=p_2=p_3=\frac{p}{3}$ in \eqref{eq:gen-pauli-channel}.
Note that $\|\id - \cD_p\|_\dn = 2p$, which can be seen as follows:
The diamond norm distance in \eqref{eq:diamond-norm}  
is at least $2p$ by choosing $M$ to be the
maximally entangled state, and at most $2p$ by 
the feasible solution $Z_{AB} = p \cJ(\id)$ 
in \eqref{eq:diamond-norm-sdp}. 
\Cref{thm:complementary-degrading} then implies that $\cD_p$ is $2^{2.5} p^{1.5}$-degradable when choosing the complementary channel $\cD_p^c$ of the depolarizing channel with error $p$ as the degrading map.
However, solving the SDP \eqref{eq:approx-deg-sdp} numerically shows that $\dg(\cD_p) 
\approx O(p^2)$, which is better than the bound promised by \Cref{thm:complementary-degrading} by half an order.

Here, we derive an analytic proof of the above numerical observation.  We will
prove that $\dg(\cD_p)$ is indeed $O(p^2)$.  This is achieved by 
choosing an approximate
degrading map to be the complementary channel $\cD_s^c$, where $s =
p+ap^2$ for a suitable choice of the parameter $a>0$.  To see the
intuition, suppose we want $\cD_s^c \circ \cD_p$ to be close to
$\cD_p^c$.  Then, choosing $s$ to be slightly larger than $p$
transmits slightly more information to Bob in the output of $\cD_s^c$,
which compensates for the slightly worse input to $\cD_s^c$ that is 
corrupted by $\cD_p$.

\begin{theorem} \label{thm:depol-degradability-parameter}
	For $p \approx 0$, we have
\begin{align}
\dg(\cD_p) \leq \frac{8}{9}\left(6 + \sqrt{2}\right) p^2 + O(p^{3}).
\end{align}
\end{theorem}

\begin{proof}
We first show that, setting $s = p + ap^2$, there is a value of $a$ for which $\|
\cD_s^c \circ \cD_p - \cD_p^c \|_\dn \approx O(p^2)$.  Then, for this
$s$, we derive two upper bounds to $\| \cD_s^c \circ \cD_p - \cD_p^c
\|_\dn$: a rough bound with a simpler derivation showing the idea, and an improved bound with a more complex
derivation yielding a better leading constant.

The complementary channel of $\cD_p$, which we refer to as the \emph{epolarizing channel} (cf.~\cite{LW15}), can be chosen to be
\begin{align}
\cD_p^c(\rho) = 
\begin{pmatrix}
(1- p) \tr (\rho) & 
\sqrt{\frac{p(1-p)}{3}} \ip{X}{\rho} & 
\sqrt{\frac{p(1-p)}{3}} \ip{Y}{\rho} & 
\sqrt{\frac{p(1-p)}{3}} \ip{Z}{\rho} \\[2mm]
\sqrt{\frac{p(1-p)}{3}} \ip{X}{\rho} & 
\frac{p}{3} \tr (\rho) &
- \frac{i p}{3} \ip{Z}{\rho} &
\frac{i p}{3} \ip{Y}{\rho} \\[2mm]
\sqrt{\frac{p(1-p)}{3}} \ip{Y}{\rho} & 
\frac{i p}{3} \ip{Z}{\rho} &
\frac{p}{3} \tr (\rho) &
- \frac{i p}{3} \ip{X}{\rho} \\[2mm]
\sqrt{\frac{p(1-p)}{3}} \ip{Z}{\rho} & 
- \frac{i p}{3} \ip{Y}{\rho} &
\frac{i p}{3} \ip{X}{\rho} &
\frac{p}{3} \tr (\rho)
\end{pmatrix}\!.
\label{eq:epolarizing-channel}
\end{align}
Using \eqref{eq:depolarizing-channel} and \eqref{eq:epolarizing-channel}, we further
obtain $\cD_s^c \circ \cD_p(\rho) = $ 
\begin{align}
\begin{pmatrix}
(1- s) \tr(\rho) & 
\!\!\! \sqrt{\frac{s(1-s)}{3}} \left( 1{-}\frac{4p}{3}\right) \ip{X}{\rho} & 
\!\!\! \sqrt{\frac{s(1-s)}{3}} \left( 1{-}\frac{4p}{3}\right) \ip{Y}{\rho} & 
\!\!\! \sqrt{\frac{s(1-s)}{3}} \left( 1{-}\frac{4p}{3}\right) \ip{Z}{\rho} \\[2mm]
\sqrt{\frac{s(1-s)}{3}} \left( 1{-}\frac{4p}{3}\right) \ip{X}{\rho} & 
\frac{s}{3}  \tr(\rho) &
- \frac{i s}{3}\left( 1{-}\frac{4p}{3}\right)\ip{Z}{\rho} &
\frac{i s}{3}\left( 1{-}\frac{4p}{3}\right) \ip{Y}{\rho} \\[2mm]
\sqrt{\frac{s(1-s)}{3}} \left( 1{-}\frac{4p}{3}\right) \ip{Y}{\rho} & 
\frac{i s}{3}\left( 1{-}\frac{4p}{3}\right) \ip{Z}{\rho} &
\frac{s}{3} \tr(\rho) &
- \frac{i s}{3} \left( 1{-}\frac{4p}{3}\right) \ip{X}{\rho} \\[2mm]
\sqrt{\frac{s(1-s)}{3}} \left( 1{-}\frac{4p}{3}\right) \ip{Z}{\rho} & 
- \frac{i s}{3}\left( 1{-}\frac{4p}{3}\right) \ip{Y}{\rho} &
\frac{i s}{3} \left( 1{-}\frac{4p}{3}\right)\ip{X}{\rho} &
\frac{s}{3} \tr(\rho) 
\end{pmatrix}\!.
\label{eq:epolarizing-channel-degraded}
\end{align}
We set $s=p+ap^2$ and $\Phi = \cD_p^c - \cD_s^c \circ \cD_p = \cD_p^c - \cD_{p+ap^2}^c\circ \cD_p$, which is given by the difference between \eqref{eq:epolarizing-channel} and \eqref{eq:epolarizing-channel-degraded}.

We first show that for some $a$, $\| \cD_s^c \circ
\cD_p - \cD_p^c \|_\dn \approx O(p^2)$, and we derive the following 
upper bound on the degradability parameter $\dg(\cD_p)$,
\begin{align}
\dg(\cD_p) \leq \frac{256}{3} p^2 + O(p^{5/2}) \,. \label{eq:depol-deg-weaker-bound}
\end{align}
To upper bound $\|\Phi\|_\dn$, we apply Lemma
\ref{lem:coefficients-bound} to the map $\Phi\colon A\to E$ with $|A|=2$ and $|E|=4$:
\begin{align}
\|\Phi\|_\dn 
\leq |A| \, |E|^2 \, \| \cJ(\Phi) \|_{\max} 
\leq 32 \; \| \cJ(\Phi) \|_{\max}.\label{eq:diamond-norm-bound}
\end{align}
Hence, we need to evaluate $\|\cJ(\Phi)\|_{\max}$, where the Choi matrix is given
by $\cJ(\Phi) = \sum_{i,j=0}^{1} |i\>\<j|$ $\ox \, \Phi(|i\>\<j|) $.
Due to the block structure of the Choi matrix, 
$\| \cJ(\Phi) \|_{\max} = \max_{i,j} \| \Phi(|i\>\<j|) \|_{\max}$.  
To find this maximum, 
first note that for any $i$ and $j$, the quantities 
$\tr (|i\>\<j|)$, 
$|\ip{X}{|i\>\<j|}|$, 
$|\ip{Y}{|i\>\<j|}|$, and $|\ip{Z}{|i\>\<j|}|$ can 
only be $0$ or $1$. 
So, from inspection of the difference between \eqref{eq:epolarizing-channel}
and \eqref{eq:epolarizing-channel-degraded}, 
$\max_{i,j} \| \Phi(|i\>\<j|) \|_{\max}$ is 
either $s-p=ap^2$, or $\frac{1}{3}|s-p+\frac{4}{3}ps| = \frac{1}{3}|a-\frac{4}{3}|p^2+O(p^3)$, or 
\begin{align}
c(p) \coloneqq 
\sqrt{\frac{p(1-p)}{3}} - \left(1-\frac{4p}{3}\right)\sqrt{\frac{(p+ap^2)(1-p-ap^2)}{3}} \,.
\label{eq:c-function}
\end{align}
Expanding $c(p)$ in a Taylor series around $p=0$ yields
\begin{align}
c(p) 
= \left(\frac{4}{3\sqrt{3}}-\frac{a}{2\sqrt{3}}\right) p^{3/2} + O(p^{5/2}) \,,
\end{align}
which is $O(p^{5/2})$ if $a=\frac{8}{3}$.
With this choice, $\max_{i,j} \| \Phi(|i\>\<j|) \|_{\max} = ap^2 = \frac{8}{3}p^2$ for 
sufficiently small $p$.  
Altogether, $\| \cJ(\Phi) \|_{\max} \leq \frac{8}{3}p^2 
+ O(p^{5/2})$, and using this with \eqref{eq:diamond-norm-bound} gives
$\| \cJ(\Phi) \|_\dn \leq \frac{256}{3}p^2 + O(p^{5/2})$, 
which completes the proof of \eqref{eq:depol-deg-weaker-bound}.

Finally, to prove the stronger assertion of the theorem, 
\begin{align}
\dg(\cD_p) \leq \frac{8}{9}\left(6 + \sqrt{2}\right) p^2 + O(p^{3}),\label{eq:depol-deg-stronger-bound}
\end{align}
we keep the choice $a=\frac{8}{3}$ to enforce that all coefficients of $\Phi = \cD_p^c - \cD_{p+ap^2}^c\circ \cD_p$ are $O(p^2)$ by Corollary \ref{cor:choi-coefficients}.
However, we upper bound $\|\Phi\|_\diamond$ with a different technique.  
Since $\Phi$ is a Hermiticity-preserving map, its diamond norm is maximized by a rank-$1$ state 
(see for example \cite{Wat16}). 
Furthermore, since $\cD_p^c$ and $\cD_{p+ap^2}^c\circ \cD_p$ are jointly covariant under the full unitary group, the diamond norm $\|\Phi\|_\diamond$ is maximized by the (normalized) maximally entangled state $\frac{1}{\sqrt{2}}(|00\rangle + |11\rangle)$ \cite[Cor.~2.5]{LKDW17}, and hence,
\begin{align}
\|\Phi\|_\diamond = \frac{1}{2} \|\cJ(\Phi)\|_1 = \frac{1}{2} \left\|\cJ(\cD_p) - \cJ(\cD_{p+ap^2}^c\circ \cD_p) \right\|_1.
\end{align}
It follows from \eqref{eq:epolarizing-channel} and \eqref{eq:epolarizing-channel-degraded} that $\frac{1}{2}\cJ(\Phi) = \begin{pmatrix}
J_{00} & J_{01} \\ J_{10} & J_{11}
\end{pmatrix},$
where
\begin{align}
J_{00} &= \begin{pmatrix}
\frac{4 p^2}{3} & 0 & 0 & \frac{1}{2} c(p) \\[0.2cm]
0 & -\frac{4 p^2}{9} & -\frac{2}{27} i p^2 (8 p-3) & 0 \\[0.2cm]
0 & \frac{2}{27} i p^2 (8 p-3) & -\frac{4 p^2}{9} & 0 \\[0.2cm]
\frac{1}{2} c(p) & 0 & 0 & -\frac{4 p^2}{9}
\end{pmatrix}\\[1ex]
J_{01} &= 
\begin{pmatrix}
0 & \frac{1}{2} c(p) & \frac{i}{2} c(p) & 0 \\[0.2cm]
\frac{1}{2} c(p) & 0 & 0 & -\frac{2}{27} p^2 (8 p-3) \\[0.2cm]
 \frac{i}{2} c(p) & 0 & 0 & -\frac{2}{27} i p^2 (8 p-3) \\[0.2cm]
0 & \frac{2}{27} p^2 (8 p-3) & \frac{2}{27} i p^2 (8 p-3) & 0
\end{pmatrix}\\[1ex]
J_{10} &= \begin{pmatrix}
0 & \frac{1}{2} c(p) & - \frac{i}{2} c(p) & 0 \\[0.2cm]
\frac{1}{2} c(p) & 0 & 0 & \frac{2}{27} p^2 (8 p-3) \\[0.2cm]
- \frac{i}{2} c(p) & 0 & 0 & -\frac{2}{27} i p^2 (8 p-3) \\[0.2cm]
0 & -\frac{2}{27} p^2 (8 p-3) & \frac{2}{27} i p^2 (8 p-3) & 0
\end{pmatrix}\\[1ex]
J_{11} &= \begin{pmatrix}
\frac{4}{3} p^2 & 0 & 0 & -\frac{1}{2} c(p) \\[0.2cm]
0 & -\frac{4}{9} p^2 & \frac{2}{27} i p^2 (8 p-3) & 0 \\[0.2cm]
0 & -\frac{2}{27} i p^2 (8 p-3) & -\frac{4}{9} p^2 & 0 \\[0.2cm]
-\frac{1}{2} c(p) & 0 & 0 & -\frac{4}{9} p^2
\end{pmatrix}\!,
\end{align}
with $c(p)$ as defined in \eqref{eq:c-function}.
Using the triangle inequality for the trace norm, we get
\begin{align}
\frac{1}{2}\|\cJ(\Phi)\|_1 \leq \|J_{00}\|_1 + \|J_{01}\|_1 + \|J_{10}\|_1 + \|J_{11}\|_1 \eqqcolon F(p).
\end{align}
A Taylor expansion shows that $F(p) = \frac{8}{9}\left(6 + \sqrt{2}\right) p^2 + O(p^{3})$, from which the bound \eqref{eq:depol-deg-stronger-bound} follows.\footnote{See the Mathematica notebook \texttt{depol-deg-bound.nb} included as an ancillary file on \url{https://arxiv.org/abs/1705.04335}.}
\end{proof}

In \Cref{fig:depolarizing-channel} we compare the optimal degradability parameter $\dg(\cD_p)$ with the quantity $\|\cD_p^c - \cD_{p+ap^2}^c\circ \cD_p \|_{\dn}$ for $a=\frac{8}{3}$, and the analytical upper bound $\frac{8}{9}(6+\sqrt{2})p^2$ obtained from Theorem \ref{thm:depol-degradability-parameter}.

\begin{figure}
	\centering
	\begin{tikzpicture}
	\begin{axis}[
	scale = 1.3,
	axis lines=left, 
	scaled ticks=false, 
	tick label style={/pgf/number format/fixed, /pgf/number format/precision=3},
	legend style = {at = {(0.05,0.95)},anchor = north west},
	legend cell align = left,
	grid = major,
	xlabel = $p$,
	xtick = {0,0.02,0.04,0.06,0.08,0.1},  
	] 
	\addplot[color=blue,thick] table[x=p,y=d] {depol.dat};
	\addplot[color=red,thick] table[x=p, y=s] {depol.dat};
	\addplot[draw=none, color=white] table[x=p, y=s] {depol.dat};
	\addplot[smooth,mark=none, color = green, thick,domain=0:0.1] {8/9*(6+sqrt(2))*x^2};
	\legend{$\dg(\cD_p)$, $\|\cD_p^c - \cD_{s}^c\circ \cD_p \|_{\dn}$, where $s = p + \frac{8}{3} p^2$,$\frac{8}{9}(6+\sqrt{2})p^2$};
	\end{axis}
	\end{tikzpicture}
	\caption{Plot of the optimal degradability parameter $\dg(\cD_p)$ (blue) of the qubit depolarizing channel $\cD_p$ computed using the SDP \eqref{eq:approx-deg-sdp}, together with the degradability parameter $\|\cD_p^c - \cD_{s}^c\circ \cD_p \|_{\dn}$ (red) obtained using the degrading map $\cD_{s}^c$ with $s = p + \frac{8}{3}p^2$, and the analytical upper bound $\frac{8}{9}(6+\sqrt{2})p^2$ (green) obtained from Theorem \ref{thm:depol-degradability-parameter}.}
	\label{fig:depolarizing-channel}
\end{figure}

Combining \Cref{thm:depol-degradability-parameter} with \Cref{lem:qbound1} and \Cref{lem:pbound1}, and using the fact that 
\begin{align}
I_c(\cD_p) = 1-h(p)-p\log 3,
\end{align} 
we obtain the following:
\begin{theorem}\label{cor:qpcapdepol}
	For small $p$, the quantum and private capacity of the qubit depolarizing channel $\cD_p$ are given by
	\begin{align}
	1 -h(p) - p\log 3 ~ \leq ~ & Q(\cD_p) ~ \leq ~ 1-h(p) - p\log 3 - \frac{16}{9}(6+\sqrt{2}) \, p^2 \log p + O(p^2)\\
	1-h(p) - p\log 3 ~ \leq ~ & P(\cD_p) ~ \leq ~ 1-h(p) - p\log 3 - \frac{16}{3}(6+\sqrt{2}) \, p^2 \log p + O(p^2).
	\end{align}
\end{theorem}

\subsection{Pauli channels and the $XZ$-channel} \label{sec:gen-pauli-channels}
The above discussion can be extended to all Pauli channels of the form in \eqref{eq:gen-pauli-channel}.  Note that $\| \cN_\bp - \id \|_\dn = 2 (p_1 + p_2 + p_3)$ by an argument similar to the one given for the depolarizing channel.  
To state our result, we consider a Pauli channel $\cN_\bp$, where the probabilities $p_i$ for $i=1,2,3$ are polynomials $p_i(p) = c_i p + d_i p^2 + \cdots$ in a single parameter $p\in[0,1]$ without constant terms, and $p_0 = 1 - p_1 - p_2 - p_3$.  
(Note that all Pauli channels can be modeled this way, and the polynomials are not unique.)  
We now define 
\begin{align}
\bs(a_1,a_2,a_3)\coloneqq (p_0', p_1(p+a_1p^2), p_2(p+a_2p^2), p_3(p+a_3 p^2)),\label{eq:s-polynomial}
\end{align}
where again $p_0' = 1 - p_1(p+a_1p^2) - p_2(p+a_2p^2) - p_3(p+a_3p^2)$.\footnote{Note that 
$p_i(p+a_ip^2)$ denote the polynomial $p_i$ with argument $p+a_ip^2$, not the product 
of $p_i$ and $p+a_ip^2$. }
We then have the following result, whose proof we give in \Cref{sec:proof-gen-pauli-channels}:
\begin{theorem}\label{thm:gen-pauli-channels}
	Let $\cN_{\bp}$ be a generalized Pauli channel with $\bp = (p_0,p_1(p),p_2(p),p_3(p))$, where $p_0 = 1-p_1(p)-p_2(p)-p_3(p)$, and the $p_i(p)$ are polynomials in $p$ with $p_i(0)=0$ for $i=1,2,3$. 
	Denote by $c_i$ the coefficient of $p$ in $p_i(p)$ for $i=1,2,3$.
	If $c_i\neq 0$ then the choices $a_i \coloneqq 4\sum_{j\neq i}c_j$ in \eqref{eq:s-polynomial} ensure that
	\begin{align}
	\| \cN_\bp^c - \cN_{\bs(a_1,a_2,a_3)}^c \circ \cN_\bp \|_\dn = O(p^2).\label{eq:quadratic-behavior}
	\end{align}
	If $c_i=0$ for some $i$, then any choice of $a_i$ for that $i$ yields \eqref{eq:quadratic-behavior}.
	Furthermore, we have 
	\begin{align}
	\dg(\cN_{\bp}) \leq \, 256 \, |c_1c_2+c_1c_3+c_2c_3| \, p^2 + O(p^3) \,.\label{eq:deg-pauli-channels}
	\end{align}
\end{theorem}

\Cref{lem:qbound1} and \Cref{lem:pbound1} now yield the following:

\begin{corollary}\label{cor:pauli-channel-capacities}
	Let $\cN_{\bp}$ be a Pauli channel as defined in \Cref{thm:gen-pauli-channels}, and define $
	C_\bp \coloneqq |c_1c_2 + c_1c_3+c_2c_3|.$
	The quantum and private capacity of $\cN_\bp$ are given by
	\begin{align}
	I_c(\cN_\bp) ~\leq~ & Q(\cN) ~\leq~ I_c(\cN_\bp) - \, 512 \; C_\bp \, p^2\log p + O(p^2)\\
	I_c(\cN_\bp) ~\leq~ & P(\cN) \, ~\leq~ I_c(\cN_\bp) - \,1536 \; C_\bp \, p^2\log p + O(p^2).
	\end{align}
\end{corollary}

\Cref{thm:gen-pauli-channels} and \Cref{cor:pauli-channel-capacities} include the (weaker) result from \Cref{sec:depol-channel} about the depolarizing channel, for which we have $c_i = \frac{1}{3}$ for $i=1,2,3$, and hence $a_i = \frac{8}{3}$ and $C_\bp = \frac{1}{3}$.

Another interesting example in the class of generalized Pauli channels is the $XZ$-channel
\begin{align}
\cN_{p,\,q}^{XZ} (\rho) \coloneqq (1-p)(1-q) \rho + p(1-q) X\rho X + pq Y\rho Y + (1-p)q Z\rho Z ,\label{eq:XZ-channel}
\end{align}
that implements independent $X$ and $Z$ errors by applying an $X$-dephasing with probability $p$, and a $Z$-dephasing with probability $q$.
For our discussion, we set $p=q$ and denote the resulting $XZ$-channel by $\cC_p$,
\begin{align}
\cC_p (\rho) = (1-p)^2 \rho + (p-p^2)X\rho X + p^2 Y\rho Y + (p-p^2)Z\rho Z. \label{eq:XZ-channel-p}
\end{align}
Thus, we have $c_1 = 1$, $c_2 = 0$, $c_3 = 1$, $d_1 = -1$, $d_2 = 1$, and $d_3 = -1$.   
Hence, the choices $a_1 = a_2 = a_3 = 4$ ensure \eqref{eq:quadratic-behavior} by \Cref{thm:gen-pauli-channels}, and from \eqref{eq:deg-pauli-channels} we obtain the analytic bound $\dg(\cC_p) \leq 256 p^2 + O(p^3)$.
However, similar to Theorem \ref{thm:depol-degradability-parameter}, we can further improve the bound on $\dg(\cC_p)$:
\begin{theorem}\label{thm:XZ-improved-coefficient}
	For $p\approx 0$, we have
	\begin{align}
	\dg(\cC_p) \leq 16 p^2 + 32 p^{5/2} + O(p^3).
	\end{align}
\end{theorem}

\begin{proof}
	For the $XZ$-channel $\cC_p$, we have $p_0 = (1-p)^2$, $p_1=p_3=p-p^2$, and $p_2=p^2$ by \eqref{eq:XZ-channel-p}.
	Furthermore, as in the discussion above we choose $s=p+4p^2$, and set $q_0 = (1-s)^2$, $q_1 = q_3 = s-s^2$, and $q_2 = s^2$, such that the map $\Phi=\cC^c_p-\cC^c_{s}\circ\cC_p$ as given in \eqref{eq:phi-matrix} has coefficients that are $O(p^2)$ by Corollary \ref{cor:choi-coefficients}.
	Since $\Phi$ is Hermiticity-preserving, its diamond norm is maximized by a pure state \cite{Wat16}, and since both $\cC^c_p$ and $\cC^c_s\circ\cC_p$ are covariant with respect to the Pauli group, a 1-design, this pure state can be chosen to be the maximally entangled state $\frac{1}{2}(|00\rangle + |11\rangle)$ \cite[Cor.~2.5]{LKDW17}.
	Hence,
	\begin{align}
	\|\cC^c_p - \cC^c_s\circ\cC_p\|_\diamond = \|\Phi\|_\diamond = \frac{1}{2} \|\cJ(\Phi)\|_1.
	\end{align}
	Exploiting the block structure $\frac{1}{2}\cJ(\Phi) = \begin{pmatrix}
	G_{00} & G_{01} \\ G_{10} & G_{11}
	\end{pmatrix}$
	that follows from \eqref{eq:phi-matrix} together with the triangle inequality for the trace norm, we obtain the upper bound
	\begin{align}
	\dg(\cC_p) \leq \|\cC^c_p - \cC^c_s\circ\cC_p\|_\diamond = \frac{1}{2} \|\cJ(\Phi)\|_1 \leq \|G_{00}\|_1 + \|G_{11}\|_1 + \|G_{01}\|_1 + \|G_{10}\|_1, \label{eq:dg-cC-upper-bound}
	\end{align}
	and a Taylor expansion of the right-hand side of \eqref{eq:dg-cC-upper-bound} shows that $\dg(\cC_p) \leq 16 p^2 + 32 p^{5/2} + O(p^3)$, which proves the claim.\footnote{We refer to the Mathematica notebook \texttt{XZ-deg-bound.nb} included in the ancillary files of \url{https://arxiv.org/abs/1705.04335} for details.}
\end{proof}

In Figure \ref{fig:XZ-channel} we compare the optimal degradability parameter $\dg(\cC_p)$ with the quantity $\|\cC_p^c - \cC_s^c\circ \cC_p \|_{\dn}$ and the analytical upper bound $16p^2 + 32p^{5/2}$ obtained from Theorem \ref{thm:XZ-improved-coefficient}.

Numerics suggest that the coherent information $I_c(\cC_p)$ is achieved by a completely mixed state $\pi = I/2$, giving
\begin{align}
I_c(\cC_p) = I_c(\pi; \cC_p) = 1-2 h(p).
\label{eq:XZ-coh-info}
\end{align}
Putting Corollary \ref{cor:pauli-channel-capacities} and Theorem \ref{thm:XZ-improved-coefficient} together, and 
using the above coherent information expression, we obtain 
\begin{corollary}
	For small $p$, the quantum and private capacity of the $XZ$-channel $\cC_p = \cN_{p,p}^{XZ}$ with equal $X$- and $Z$-dephasing probability are given by
	\begin{align}
	1 - 2 h(p) ~\leq~ & Q(\cC_p) ~\leq~ 1 - 2 h(p) - 32 \, p^2 \log p + O(p^2)\\
	1 - 2 h(p) ~\leq~ & P(\cC_p) \, ~\leq~ 1 - 2 h(p) - 96 \, p^2 \log p + O(p^2).
	\end{align}
\end{corollary}

\begin{figure}
	\centering
	\begin{tikzpicture}
	\begin{axis}[
	scale = 1.3,
	axis lines=left, 
	scaled ticks=false, 
	tick label style={/pgf/number format/fixed, /pgf/number format/precision=3},
	legend style = {at = {(0.05,0.95)},anchor = north west},
	legend cell align = left,
	grid = major,
	xlabel = $p$,
	xtick = {0,0.02,0.04,0.06,0.08,0.1}, 
	] 
	\addplot[color=blue,thick] table[x=p,y=d] {XZ.dat};
	\addplot[color=red,thick] table[x=p, y=s] {XZ.dat};
	\addplot[draw=none, color=white] table[x=p, y=s] {XZ.dat};
	\addplot[smooth,mark=none, color = green, thick,domain=0:0.1] {16*x^2+32*x^(2.5)};
	\legend{$\dg(\cC_p)$, $\|\cC_p^c - \cC_{s}^c\circ \cC_p \|_{\dn}$, where $s = p + 4 p^2$,$16p^2 + 32p^{5/2}$};
	\end{axis}
	\end{tikzpicture}
	\caption{Plot of the optimal degradability parameter $\dg(\cC_p)$ (blue) of the $XZ$-channel $\cC_p$ computed using the SDP \eqref{eq:approx-deg-sdp}, together with the degradability parameter $\|\cC_p^c - \cC_{s}^c\circ \cC_p \|_{\dn}$ (red) obtained using the degrading map $\cC_{s}^c$ with $s = p + 4p^2$, and the analytical upper bound $16p^2 + 32p^{5/2}$ (green) obtained from Theorem \ref{thm:XZ-improved-coefficient}.}
	\label{fig:XZ-channel}
\end{figure}

\section{Conclusion}\label{sec:conclude}

Our results can be extended to cover \emph{generalized low-noise} channels $\cN$, for which there exists another channel $\cM$ such that $\| \cM\circ\cN - I \|_{\diamond} \leq \epsilon.$
For example, this includes all channels that are close to isometric channels.
For a generalized low-noise channel, we have by Theorem \ref{thm:complementary-degrading} that
\begin{align} \label{Eq:MN-approx-deg}
\| (\cM\circ\cN)^{c} - (\cM\circ\cN)^{c}\circ (\cM\circ\cN) \|_{\diamond} \leq 2 \epsilon^{3/2}.
\end{align}
Furthermore, up to an isometry, 
\begin{align}
(\cM\circ \cN)^c(\rho) = (\cM^c\otimes I_{E_1})(U_\cN \rho U_{\cN}^\dagger),
\end{align}
where $U_\cN: A \rightarrow BE_1$ is an isometric extension of $\cN$ and $\cM^c : B \rightarrow E_2$, so that $\tr_{E_2}(\cM\circ \cN)^c(\rho) = \cN^{c}(\rho)$.
Equation \ref{Eq:MN-approx-deg} therefore implies
\begin{align}
\| \cN^{c} - \tr_{E_2}(\cM\circ\cN)^{c}\circ (\cM\circ\cN) \|_{\diamond} \leq 2 \epsilon^{3/2},
\end{align}
so that letting $\cD =  \tr_{E_2}(\cM\circ\cN)^{c}\circ \cM$ we have $\| \cN^{c} - \cD\circ\cN \|_{\diamond} \leq 2 \epsilon^{3/2}$ and $\cN$ has degradability parameter no bigger than $ 2 \epsilon^{3/2}$.  
We thus find that the same bounds as in Theorem \ref{lem:qbound2} apply in the case of a generalized low-noise channel $\cN$.

We conclude with some implications of our results.  
The quantum and private capacities of a quantum channel are intractable to calculate in general.  This is because nonadditivity effects require us to regularize the coherent information and private information to obtain the quantum and private capacity, respectively. 
We have shown that for low-noise channels, for which $\|\cN - \id\|_\diamond \leq \eps$, such nonadditivity effects are negligible.  In particular, we find that both the private and quantum capacities of these channels are given by the one-shot coherent information $I_c(\cN)$, up to corrections of order $\eps^{1.5}\log \eps$.  
Furthermore, for the qubit depolarizing channel $\cD_p$ we obtain even tighter bounds for both the quantum and private capacities:
\begin{align}
1 -h(p) - p\log 3 ~ \leq ~ & Q(\cD_p) ~ \leq ~ 1-h(p) - p\log 3 - \frac{16}{9}(6+\sqrt{2}) \, p^2 \log p + O(p^2)\\
1-h(p) - p\log 3 ~ \leq ~ & P(\cD_p) ~ \leq ~ 1-h(p) - p\log 3 - \frac{16}{3}(6+\sqrt{2}) \, p^2 \log p + O(p^2),
\end{align}
and similar results hold for all generalized Pauli channels in dimension $d$. 
Our key new finding is that channels within $\eps$  of perfect are also exceptionally close to degradable, with degradability parameter of $O(\eps^{1.5})$ in general and $O(p^2)$ for generalized Pauli channels.  

The nonadditivty of coherent information for a general channel implies
that degenerate codes are sometimes needed to achieve the quantum
capacity \cite{DiVincenzoSS98,SS07,Fern08,SY08,SS09,Cubitt2015unbounded}, 
but little is known about these codes despite 20 years of research.
Having shown that the coherent information is essentially the quantum
capacity for low-noise channels, we have also arrived at a refreshing
result that using random block codes on the typical subspace of the 
optimal input (for the 1-shot coherent information) essentially 
achieves the capacity.

Likewise, our findings have implications in quantum cryptography.  In quantum
key distribution, quantum states are transmitted through
well-characterized noisy quantum channels that are subject to further
adversarial attacks.  Parameter estimation is used to determine the
effective channel (see for example \cite{RGK05}) and the optimal key
rate of one-way quantum key distribution is the private capacity of
the effective channel.
These effective channels typically have low noise (e.g., $1-2\%$ in
\cite{QBER09}), and our results imply near-optimality of the simple
classical error correction and privacy amplification procedures that
approach the one-shot coherent information of the effective channel.
In particular, for the XZ-channel with bit-flip probability $p$, the
optimal key rate is $1-2h(p) + O(p^2\log p)$.

\section{Acknowledgements}

We thank Charles Bennett, Ke Li, John Smolin, and John Watrous for inspiring discussions, Mark M.~Wilde for helpful feedback, and Min-Hsiu Hsieh and Yen-Chi Lee for pointing out a small glitch in proving the rough bounds in Theorems \ref{thm:depol-degradability-parameter} and \ref{thm:gen-pauli-channels}, which led us to prove a stronger version of Lemma \ref{lem:coefficients-bound}.
DL is further supported by NSERC and CIFAR, and FL and GS by the National Science Foundation under Grant Number 1125844.

\appendix

\section{Proof of \texorpdfstring{\Cref{thm:gen-pauli-channels}}{Theorem \ref{thm:gen-pauli-channels}}}\label{sec:proof-gen-pauli-channels}

In this appendix we prove \Cref{thm:gen-pauli-channels}, which states that for any generalized Pauli channel 
\begin{align}
\cN_\bp (\rho) = p_0 \rho + p_1 X\rho X + p_2 Y\rho Y + p_3 Z\rho Z,\label{eq:pauli-channel}
\end{align}
for which the coefficients $p_i(p) = c_i p + d_i p^2 + \cdots$ are polynomials with zero constant term in a parameter $p$, there is a choice of $a_1, a_2, a_3$ such that
\begin{align}
\| \cN_\bp^c - \cN_{\bs(a_1,a_2,a_3)}^c \circ \cN_\bp \|_\dn = O(p^2).\label{eq:quadratic-behavior-app}
\end{align}
Here, we define $\bs(a_1,a_2,a_3)\coloneqq (p_0', p_1(p + a_1 p^2), p_2(p+a_2 p^2), p_3(p+a_3p^2))$, with $p_0'$ chosen such that $\bs(a_1,a_2,a_3)$ is a probability distribution.

An isometric extension of the Pauli channel is given by
\begin{align}
V_\bp \coloneqq \sqrt{p_0}\,\one \ox |0\rangle_E + \sqrt{p_1}\,X\ox |1\rangle_E + \sqrt{p_2}\, Y \ox |2\rangle_E + \sqrt{p_3}\, Z\ox |3\rangle_E,
\end{align}
where $\lbrace |0\rangle_E, |1\rangle_E, |2\rangle_E, |3\rangle_E\rbrace$ is an orthonormal basis for the environment $E$.
For an arbitrary linear operator $\rho$, it is straightforward to find that (see e.g.~\cite{LW15}) 
\begin{align} 
\cN_\bp^c(\rho) = \tr_B(V_\bp \rho V_\bp^\dagger) = 
\begin{pmatrix}
p_0 \tr(\rho) & \sqrt{p_0 p_1} \ip{X}{\rho} & \sqrt{p_0 p_2} \ip{Y}{\rho} & \sqrt{p_0 p_3} \ip{Z}{\rho}\\[0.2cm]
\sqrt{p_0 p_1} \ip{X}{\rho} & p_1 \tr(\rho) & -i \sqrt{p_1 p_2} \ip{Z}{\rho} & i \sqrt{p_1 p_3} \ip{Y}{\rho} \\[0.2cm]
\sqrt{p_0 p_2} \ip{Y}{\rho} & i \sqrt{p_1 p_2} \ip{Z}{\rho} & p_2 \tr(\rho) & -i \sqrt{p_2 p_3} \ip{X}{\rho} \\[0.2cm]
\sqrt{p_0 p_3} \ip{Z}{\rho} & -i \sqrt{p_1 p_3} \ip{Y}{\rho} & i \sqrt{p_2 p_3} \ip{X}{\rho}  & p_3 \tr(\rho)
\end{pmatrix}. 
\end{align}
Writing $\bq=(q_0,q_1,q_2,q_3)$, the action of the map $\Phi \coloneqq \cN_\bp^c - \cN_{\bq}^c\circ \cN_\bp$ on a linear operator $\rho$ is
\begin{align}
\Phi(\rho) &= \begin{pmatrix}
(p_0-q_0) \tr(\rho) & t_1 \ip{X}{\rho} & t_2 \ip{Y}{\rho} & t_3 \ip{Z}{\rho} \\[0.2cm]
t_1 \ip{X}{\rho} & (p_1-q_1)\tr(\rho) & -i u_3 \ip{Z}{\rho} & i u_2 \ip{Y}{\rho}\\[0.2cm]
t_2 \ip{Y}{\rho} & i u_3 \ip{Z}{\rho} & (p_2-q_2)\tr(\rho) & -i u_1 \ip{X}{\rho}\\[0.2cm]
t_3 \ip{Z}{\rho} & -i u_2 \ip{Y}{\rho} & i u_1 \ip{X}{\rho} & (p_3-q_3)\tr(\rho)
\end{pmatrix}\label{eq:phi-matrix}
\end{align}
with the following coefficients:
\begin{align}
t_1 &= \sqrt{p_0p_1} - \sqrt{q_0q_1}\,(p_0+p_1-p_2-p_3) \label{eq:t1}\\
t_2 &= \sqrt{p_0p_2} - \sqrt{q_0q_2}\,(p_0-p_1+p_2-p_3) \label{eq:t2}\\
t_3 &= \sqrt{p_0p_3} - \sqrt{q_0q_3}\,(p_0-p_1-p_2+p_3) \label{eq:t3}\\
u_1 &= \sqrt{p_2p_3} - \sqrt{q_2q_3}\,(p_0+p_1-p_2-p_3) \label{eq:u1}\\
u_2 &= \sqrt{p_1p_3} - \sqrt{q_1q_3}\,(p_0-p_1+p_2-p_3) \label{eq:u2}\\
u_3 &= \sqrt{p_1p_2} - \sqrt{q_1q_2}\,(p_0-p_1-p_2+p_3) \label{eq:u3}.
\end{align}
Similar to the case for the depolarizing channel, $\| \cJ(\Phi) \|_{\max}$
is the maximum among $|p_i - q_i|$ for $i=0,1,2,3$ and $t_i, u_i$ for
$i=1,2,3$.

We choose $q_i(p) = p_i(p+a_ip^2)$ for $i=1,2,3$, and $q_0 = 1-q_1-q_2-q_3$.
It follows that $p_i-q_i = O(p^2)$ for $i=0,1,2,3$.
Similarly, expanding the $u_i$ in a Taylor series around $p=0$ shows that 
$u_i =O(p^2)$ for any choices of $a_i$.
Hence, it remains to check the coefficients $t_i$, which we again expand in a Taylor series around $p=0$:
\begin{align}
t_1 &= \frac{\sqrt{c_1}}{2}(a_1 - 4(c_2+c_3)) p^{3/2} + O(p^{5/2})\\
t_2 &= \frac{\sqrt{c_2}}{2}(a_2 - 4(c_1+c_3)) p^{3/2} + O(p^{5/2})\\
t_3 &= \frac{\sqrt{c_3}}{2}(a_3 - 4(c_1+c_2)) p^{3/2} + O(p^{5/2}).
\end{align}
We see that if $c_i=0$, then the coefficient of $p^{3/2}$ in $t_i$ vanishes 
for any choice of $a_i$.
If on the other hand $c_i\neq 0$, then the coefficient of $p^{3/2}$ in $t_i$ vanishes upon choosing $a_i = 4\sum_{j\neq i} c_j$.

In summary, for the above choices of $a_1, a_2, a_3$, the max norm $\| \cJ(\Phi) \|_{\max}$ is the
maximum of $|p_i-q_i|$ for $i=0,1,2,3$ and $u_i$ for $i=1,2,3$.  
Recall that
\begin{align}
p_i(p) &= c_i p + d_i p^2 + O(p^3)\\
q_i(p) &= p_i(p+a_ip^2) = c_ip + (a_ic_i + d_i) p^2 + O(p^3).
\end{align}
Hence, since $p\approx 0$,
\begin{align}
\max_i |p_i-q_i| &= |p_0 - q_0|\\
 &= \left| \sum (p_i - q_i) \right|\\
 &= \left|\sum a_ic_i \right| p^2 + O(p^3)\\
 &= 8 |c_1c_2 + c_1c_3 + c_2c_3| p^2 + O(p^3),\label{eq:max-p-minus-q}
\end{align}
where we substituted $a_i = 4\sum_{j\neq i} c_j$ in the last line.
Expanding $u_i$ for $p \approx 0$ gives 
\begin{align}
\begin{aligned}
u_1 &= -4 c_1 \sqrt{c_2 c_3} p^2 + O(p^3)\\
u_2 &= -4 c_2 \sqrt{c_1 c_3} p^2 + O(p^3)\\
u_3 &= -4 c_3 \sqrt{c_1 c_2} p^2 + O(p^3). 
\end{aligned} 
\label{eq:u-expanded}
\end{align}
From \eqref{eq:max-p-minus-q} and \eqref{eq:u-expanded}, we infer
\begin{align}
\| \cJ(\Phi) \|_{\max} = |p_0 - q_0| = 
8 |c_1 c_2 + c_1 c_3 + c_2 c_3| p^2 + O(p^3),
\end{align}
so that by Lemma \ref{lem:coefficients-bound} we have 
$\| \Phi \|_\dn \leq 256 |c_1 c_2 + c_1 c_3 + c_2 c_3| p^2 + O(p^3)$, 
which concludes the proof of \Cref{thm:gen-pauli-channels}.

\printbibliography[title={References},heading=bibintoc]

\end{document}